\documentclass[10pt,oneside]{article}
\usepackage{amsfonts,amssymb,graphicx}
\usepackage{epstopdf}
\usepackage{cite}
\usepackage{indentfirst}

\usepackage{amsthm}
\usepackage{amsmath}
\usepackage{latexsym}
\usepackage{amsthm}

\usepackage{amsmath, amssymb, mathrsfs, dsfont, mathtools, empheq,bm,bbm}
\usepackage{color}

\def\be{\begin{equation}}
\def\ee{\end{equation}}

\newtheorem{theorem}{Theorem}[section]
\newtheorem{prop}[theorem]{Proposition}
\newtheorem{lemma}{Lemma}
\newtheorem{remark}{Remark}
\newtheorem{definition}{Definition}[section]

\newcommand{\PP}{\mathbb{P}}
\newcommand{\R}{\mathbb{R}}
\newcommand{\E}{\mathbb{E}}
\newcommand{\N}{\mathbb{N}}
\newcommand{\balpha}{\bm{\alpha}}
\newcommand{\bbeta}{\bm{\beta}}

\newcommand{\bsigma}{\bm{\sigma}}

\begin{document}

\begin{center}
\vspace{1truecm}
{\bf\Large A Multi-Scale Spin-Glass Mean-Field Model}\\
\vspace{1cm}
{Pierluigi Contucci, Emanuele Mingione}\\
\vspace{.5cm}
{\small Dipartimento di Matematica} \\
{\small Universit\`a di Bologna, 40127 Bologna, Italy}\\
\vskip 1truecm
\end{center}
\vskip 1truecm
\begin{abstract}\noindent

\noindent In this paper a multi-scale version of the Sherrington and Kirkpatrick model is introduced and studied. The pressure per particle in the thermodynamical limit is proved to obey a variational principle of Parisi type. The result is achieved by means of lower and upper bounds. The lower bound is obtained with a Ruelle cascade using the interpolation technique, while the upper bound exploits factorisation properties of the equilibrium measure and the synchronisation technique.
\end{abstract}


%

\section{The Multiscale SK model}

The equilibrium statistical mechanics of a general disordered system can be described in between two prescriptions known, in the literature, as quenched and annealed. The spin-glass phase, for instance, is described by the quenched measure where the random coupling disorder is kept fixed while the spins are thermalised according to the Boltzmann distribution. This perspective is considered physically relevant because the relaxation time of the disorder interaction variables is much slower than the one for the spin variables. Conversely in the annealed prescription the disorder variables thermalise together with the field ones. In a paper by Talagrand \cite{tal} on mean field spin glasses it was shown how to define a generalised equilibrium measure depending on a real positive number $\zeta$ with thermodynamic pressure
\be\label{tala}
P \; = \; \frac{1}{\zeta}\log \E Z_J^\zeta \; ,
\ee
where $Z_J$ is the partition function, a random variable depending on the disorder $J$ obtained integrating on the spins. The origins of this description are to be found on the replica approach to spin glasses \cite{MPV} where $\zeta$ is an integer. In \cite{BGM} $\zeta$ was treated as a {\it scale} parameter in the unit interval to interpolate a general disordered system from the quenched case, obtained when $\zeta\to 0$, to the annealed one reached at $\zeta \to 1$.

In this paper we generalise the idea of \eqref{tala} and consider a multi-scale equilibrium measure obtained by successive independent integration on suitably defined Gaussian couplings. The idea to study a system at different energy scales is common in mathematical-physics at least since the early days of the Euclidian approach to renormalisation group in quantum field theory (see \cite{Pol,Gal}). Recalling the basic concepts, a single scale model is defined as in \eqref{tala}. For two scales
$\zeta_0$ and $\zeta_1$ the model is defined in terms of an interaction $J=(J_0,J_1)$ with independent components:
\be
e^{\zeta_1 P^{(0)}} \; = \;  \E_1 Z_J^{\zeta_1} \; ,
\ee
and
\be
e^{\zeta_0 P} \; = \;  \E_0\, e^{\zeta_0P^{(0)}} \; .
\ee
For $r$ scales $\zeta_0< \zeta_1<\ldots<\zeta_{r-1} < \zeta_{r}\,=\,1$
the recursion relations are
\be\label{hier}
e^{\zeta_{l} P^{(l-1)}} \; = \;  \E_l\, e^{\zeta_lP^{(l)}} \; ,
\ee
where $0 \le l \le r$, $\E_r\,e^{P^{(r)}}=Z_J$ and $P=P^{(-1)}$.

The use of a multi-scale decomposition structure in the spin-glass problem made its first appearance with the celebrated work by Guerra on the Sherrington-Kirkpatrick model \cite{broken} where the covariance of a one-body exactly solvable system is split in many layers. The same method was later used by Talagrand in his complete proof, the upper bound, of the Parisi formula for the free energy density of the model \cite{tala2}. The idea to use the recursive structure introduced above appeared also in the theoretical physics literature. In \cite{Mon} the author uses it to investigate the properties of metastable states in a glassy system. In \cite{cuku,cuku2} the authors introduce a multi-bath equilibrium showing that it can be used to describe the correlations and response functions for a class of dynamical systems in the limit of small entropy production.

What we propose here is a generalised mean-field model where a multi-scale structure is part of the model itself and involves the interacting covariance.

A fundamental tool throughout this work, that we will use to study the multi-scale equilibrium measure defined by \eqref{hier}, are the Ruelle Probability Cascades (RPC) \cite{Ruelle, Derrida} whose use is consolidated in the spin-glass literature \cite{Bol, BolKistler, Bovierku, pan0}. A short appendix on RPC is provided at the end to make this work self-contained.

The main definitions follow.

Given $N \geq 1$ let us consider a  system  of $N$ spins $\sigma=(\sigma_i)_{\,i\leq N }\in\Sigma_{N}=\{-1,1\}^N$.  Fix an integer $r\geq1$ and  denote by $\balpha\in\N^r$  an additional degree of freedom. A configuration of the system is

\be\label{conf}
\bsigma\,=\,(\sigma,\balpha)\in \Sigma_N\times\N^r\,\equiv\,\Sigma_{N,r}
\ee

Consider a sequence  $\zeta=(\zeta_l)_{l\leq r}$ such that

\be\label{zetaseq}
0=\zeta_{-1}<\zeta_0< \zeta_1<\ldots<\zeta_{r-1} < \zeta_{r}\,=\,1
\ee

and let $(\nu_{\balpha})_{\balpha\in\N^{r}}$ be the random weights of Ruelle Probability Cascade associated to the sequence $\zeta$ (see Appendix \ref{RPC}). For $\balpha,\bbeta\in \N^r$ we denote

\be\label{wegge}
\balpha\wedge\bbeta\, = \min\,\{0\leq l \leq r\,|\,\alpha_1 = \beta_1,\ldots , \alpha_l = \beta_l, \alpha_{l+1} \neq \beta_{l+1}\}
\ee

where $\balpha\wedge\bbeta=r$ if $\balpha=\bbeta$. It's useful to think $\N^r $ as the set of leaves of an infinite tree $\mathcal{A}=\N^0\cup\N\cup\N^2\ldots\cup \N^r $  of depth $r$ and root
$\N^0=\{\emptyset\}$. Then  $\balpha\wedge\bbeta$ denotes  the level of their common ancestor, see \eqref{dwedge}.

Fix a sequence $\gamma=(\gamma_l)_{l\leq r}$ such that

\be\label{gammaseq}
0=\gamma_0<\gamma_1<\ldots<\gamma_{r} < \infty
\ee

and let $\Big(g(\balpha)\Big)_{\balpha\in\N^r}$ be a family of centered gaussian random variables with covariance
\be
\E \,g(\balpha^1)\,g(\balpha^2) = \,\left(\gamma_{\balpha^1\wedge\balpha^2}\right)^2\,
\ee

Consider a gaussian process $H_N$ on $\Sigma_{N,r}$ defined by

\be\label{hamrpc}
H_{N}(\bsigma)\,=\,\frac{1}{\sqrt{N}}\,\sum_{i,j=1}^N\,g_{ij}(\balpha)\,\sigma_i\sigma_j
\ee

where $\bsigma=(\sigma,\balpha)\in\Sigma_N\times \N^r$  and  $\left(g_{ij}(\balpha)\right)_{i,j=1,\ldots ,N}$ is a family of i.i.d. copies of $g(\balpha)$.\\

Given two configurations $\bsigma^1=(\sigma^1,\balpha^1),\,\bsigma^2=(\sigma^2,\balpha^2)\in\Sigma_{N,r}$  the covariance of the process $H_{N}$ is

\be\label{covH}
\E \,{H}_N(\bsigma^1)\,{H}_N(\bsigma^2) \,=\, N\,\left(c_{N,\gamma}(\bsigma^1,\bsigma^2)\right)^2
\ee

where
\be\label{overl}
c_{N,\gamma}\left(\bsigma^1,\bsigma^2\right)\,=\,{\gamma}_{\balpha^1\wedge\balpha^2}\,q_N(\sigma^1,\sigma^2)
\ee

and
\be\label{overlap}
q_N(\sigma^1,\sigma^2)\,=\,\frac{1}{N}\sum^N_{i=1}\,\sigma^1_i\,\sigma^2_i
\ee

is  the usual  \textit{overlap} between two configurations $\sigma^1,\sigma^2\in\Sigma_N$.
Notice that  $q_N(\sigma^1,\sigma^2)\in[-1,1]$ and $\balpha\wedge\balpha=r$ imply that   $c_{N,\gamma}(\bsigma^1,\bsigma^2)\in[-\gamma_r,\gamma_r]$.\\

We denote by $\beta=(\zeta,\gamma)$ the couple of sequences in \eqref{zetaseq} and \eqref{gammaseq}. Given $\beta$ we by $p_N(\beta)$ the  \textit{quenched pressure density } of the Multiscale SK model, defined as

\be\label{multiscalep}
p_N(\beta)= \frac{1}{N}\,\E\,\log Z_{N}(\beta)
\ee

where
\be
Z_ {N}(\beta)\,=\,\sum_{\bsigma\in\Sigma_{N,r}}\, \nu_{\balpha}\,e^{\,H_N(\bsigma)}
\ee

We notice that $p_N(\beta)$ in \eqref{multiscalep} can be also defined recursively. Let  $H_N(\sigma,l)$ be a gaussian process on  $(\sigma,l)\in\Sigma_N\times \{1,\ldots,r\}$ with covariance

\be\label{coint}
\E \,H_N(\sigma^1,l) H_N(\sigma^2,l') \,=\, N\,\delta_{l,l'}\,\left(\sqrt{\,\gamma^2_l-\gamma^2_{l-1}\,}\, q_N(\sigma^1,\sigma^2)\right)^2\,
\ee

Then, by  the property \eqref{rpcinv} of the RPC, it holds
\be\label{multiscalep2}
p_N(\beta)= \frac{1}{N}\,\log Z_{0,N}(\beta)
\ee

where  $Z_{0,N}$ is obtained recursively in the following way. We denote
by $\E_l$ denotes the average w.r.t. the randomness in $H_N(\sigma,l+1)$  and starting from

\be
Z_ {r,N}(\beta)\, =\,\sum_{\sigma}\, \prod_{1\leq l\leq r}\,e^{\,H_N(\sigma,l)}
\ee

we define
\be\label{recZ}
Z^{\zeta_{l-1}}_{l-1,N}\,=\,\E_{l-1}\,Z^{\zeta_{l-1}}_{l,N}
\ee
for any $0\leq l\leq r-1$.\\

For $r=1$ and a generic $\zeta_0$ the model was studied and solved by Talagrand in \cite{tal}. If $\zeta_0 \to 0$ we recover the SK model at inverse temperature $\gamma_1$.\\

\section{Main result}

The quenched pressure density $p_N$  in \eqref{multiscalep} is completely determined by the choice of  $\beta =(\zeta,\gamma)$. From now on  $r$ denotes the integer that defines the sequences $\zeta$ and $\gamma$
in \eqref{zetaseq} and \eqref{gammaseq}.\\

Consider an arbitrary integer $k\geq r$ and a sequence $\xi=(\xi_j)_{j\leq k}$ such that
\be\label{newseq}
0=\xi_{-1}<\xi_{0} < \xi_{1}< \ldots < \xi_k\,=\,1\
\ee

Moreover we assume that
\be\label{spacecon}
\zeta\subseteq\xi
\ee

It's useful to think $\xi$ as a the image of some  discrete distribution function. In other words given an arbitrary sequence $c=(c_j)_{j\leq k}$ such that

\be\label{newseq2}
0<c_{0} < c_{1}< \ldots < c_k<\infty
\ee

we say that a random variable $C$ taking values on the set $c$ has distribution $\xi$ if

\be\label{pdis}
\PP(C\,=\,c_j)\,=\,\xi_j- \xi_{j-1}
\ee

for any $j\leq k$. Any couple of sequences $(\xi,c)$ satisfying \eqref{newseq} and \eqref{newseq2} combined with the relation \eqref{pdis} determines an element of $\mathcal{M}[0,c_k]$ where $\mathcal{M}[0,c_k]$ denotes the set of all  distribution functions on $[0,c_k]$. However in our case the additional condition \eqref{spacecon} implies that we look at a particular subset of $\mathcal{M}[0,c_k]$.

\begin{definition}\label{funspace}
We denotes by  $\mathcal{M}_\zeta[0,c_k]$ the set of all  distribution function  $F$   on $[0,c_k]$ such that
the sequence $\zeta$ is contained in the image of $F$.
\end{definition}

Notice that if $F$ is a discrete distribution   on $[0,c_k]$  then it can be identified with a couple
$(\xi,c)$ satisfying \eqref{newseq} and \eqref{newseq2} and  the above definition implies that

\be\label{bf}
F\in \mathcal{M}_\zeta[0,c_k]\,\Leftrightarrow\, \zeta\subseteq \xi
\ee

Now given the sequence $\xi$ in \eqref{newseq} satisfying \eqref{spacecon}   consider  the following  subset of $\{0,\ldots,k\}$

\be\label{mergeset}
K_l\,=\,\left\{\,j\,:\, \zeta_{l-1}< \xi_j\leq\zeta_{l},\,0\leq j\leq k\right\}
\ee

for any $l\leq r$. Given the sequence $\gamma$ in \eqref{gammaseq} we construct a new sequence $\widetilde{\gamma}=(\widetilde{\gamma}_j)_{ j\leq k}$ defining for any  $j\leq k$

\be\label{seqq}
\widetilde{\gamma}_j\,=\,\gamma_l\,\,\,\, \mathrm{if}\,\,\, j\in K_l
\ee

We also introduce  an arbitrary sequence $q=(q_j)_{j\leq k}$ such that

\be\label{seqc}
0\,=\,q_0  \leq q_{1} \leq \ldots \leq q_{k}=1
\ee

\begin{definition}\label{spacef}
We denote $X_{\beta}$ the set of all  $x=(\xi,\widetilde{\gamma},q)$  such that
$\xi$ satisfies \eqref{newseq} and \eqref{spacecon} while  $\widetilde{\gamma}$ and $q$ are defined in
\eqref{seqq}  and \eqref{seqc} respectively.
\end{definition}

Given $x=(\xi,\widetilde{\gamma},q)\in X_\beta$,  consider the  sequence $c=(c_j)_{j\leq k}$ 
where $c_j=\widetilde{\gamma}_j\,q_j$ for any $j\leq k$. Then, from a physical point of view,   the couple $(\xi,c)$ associated to a suitable $x\in X_{\beta}$,  represents the distribution of the overlap $c_N$ in \eqref{overl} w.r.t. the  Gibbs measure in the thermodynamic limit.

Let $(J_j)_{1\leq j\leq k}$ be a collection of i.i.d. standard gaussian random variables and define
\be
Z_{k}\,=\,2\cosh\left(\sqrt{2}\,\sum_{1\leq j\leq k}\,J_p\,\Big(\widetilde{\gamma}_j^2 q_j- \widetilde{\gamma}_{j-1}^2q_{j-1}\Big)^{1/2}\right)
\ee

and recursively for $0\leq j \leq k-1$

\be\label{recZ}
Z^{\xi_{j-1}}_{j-1}\,=\,\E_{j-1}\,Z^{\xi_{j-1}}_{j}
\ee

where $\E_j$ denotes the average w.r.t. $J_{j+1}$.\\

For any $x\in X_{\beta}$ we define the Parisi functional for the Multiscale SK model the quantity

\be\label{parisifun}
\mathcal{P}_{\beta}(x)\,=\,\log\,Z_0\,-\,\frac{1}{2}\,\sum_{0\leq j\leq k-1}\,\xi_j\,\Big(
(\widetilde{\gamma}_{j+1}q_{j+1})^2-(\widetilde{\gamma}_{j}q_{j})^2\Big)
\ee

Using \eqref{rpcinv} one can prove that the Parisi functional
\eqref{parisifun} has another  useful representation.  Let $(\nu_{\balpha} )_{\balpha\in \N^k}$ the random weights of the RPC with parameter $\xi$. Consider  two independent gaussian process $z,y$  indexed by $\balpha\in \N^k$ with covariances
\begin{eqnarray}\label{htree}
 \E \,z(\balpha^1)\,z(\balpha^2) &=& 2\,(\widetilde{\gamma}_{\balpha^1\wedge\balpha^2})^2\,q_{\balpha^1\wedge\balpha^2} \\
 \E \,y(\balpha^1)\,y(\balpha^2) &=& \big( \widetilde{\gamma}_{\balpha^1\wedge\balpha^2}\,q_{\balpha^1\wedge\balpha^2}\big)^2
\end{eqnarray}

Hence it holds

\be\label{parisifun2}
\mathcal{P}_{\beta}(x)\,=\, \E \log\,\sum_{\balpha\in\N^k}\,\nu_{\balpha}\,2\cosh z(\balpha)\,-\,\E \log\,\sum_{\balpha\in\N^k}\,\nu_{\balpha}\,2\exp y (\balpha)
\ee

The main result of this work is the following

\begin{theorem}\label{upbound}
The thermodynamic limit of the quenched pressure density of the Multiscale SK model $p_N(\beta)$ in
\eqref{multiscalep} exists and is given by
\be\label{main}
\lim_{N\to\infty}p_{N}(\beta)\,= \,\inf_{x\in X_\beta}\,\mathcal{P}_{\beta}(x)
\ee
\noindent
where $\mathcal{P}_{\beta}(x)$ is the  \textit{Parisi-like } functional defined in \eqref{parisifun}
 and the set  $X_\beta$ is defined in \eqref{spacef}.
\end{theorem}

The existence of the thermodynamic limit of $p_N(\beta)$ can be proved regardless of \eqref{main} using a Guerra-Toninelli argument \cite{gt}. Indeed the covariance $c_{N,\gamma}$ in \eqref{overl}depends on $N$ only trough the overlap $q_N$ in  \eqref{overlap}, namely the covariance of an  SK model.

It would be interesting to see if the functional $\mathcal{P}_{\beta}$ is convex as it has been proved in the case of the SK model  \cite{AU}.

Notice also that in Talagrand's paper \cite{tal} where the case $r=1,\zeta_0\in(0,1)$ is considered, the  trial RPC starts from $\xi_0=\zeta_0$. Even if this requirement is not present explicitly in the definition \eqref{parisifun} for the trial functional $\mathcal{P}_{\beta}$, it's possible to show that condition \eqref{seqq} implies it.

\section{Upper bound, Guerra's interpolation}

In this section we give an upper bound for the quenched pressure of the Multiscale SK model $p_N$ defined in \eqref{multiscalep}. In the proof  given here we use  RPC formalism. The same result can be obtained working with the  recursive definition  \eqref{multiscalep2} for $p_N(\beta)$ and applying Guerra's methods \cite{broken,leshouches}.

\begin{prop}

The quenched pressure density of the Multiscale SK model $p_N(\beta)$ satisfies

\be\label{mainl}
\limsup_{N \to \infty} p_{N}(\beta)\,\leq \,\inf_{x\in X_\beta}\,\mathcal{P}_{\beta}(x)
\ee
where the functional $\mathcal{P}_{\beta}(x)$ and the set $X_\beta$ are  defined in  \eqref{parisifun} and \eqref{spacef} respectively.
\end{prop}

\begin{proof}

Let $(\nu_{\balpha} )_{\balpha\in \N^{k}}$ the random weights of the RPC with parameter $\xi=(\xi_j)_{j\leq k}$  in \eqref{newseq} and consider two independent gaussian process $\widetilde{g},z$  indexed by $\balpha\in \N^k$ with covariances

\begin{eqnarray}
 \E \,\widetilde{g}(\balpha^1)\,\widetilde{g}(\balpha^2) &=&\ \big(\widetilde{\gamma}_{\balpha^1\wedge\balpha^2}\big)^2\label{covgt}\\
 \E \,z(\balpha^1)\,z(\balpha^2) &=& 2\,(\widetilde{\gamma}_{\balpha^1\wedge\balpha^2})^2 q_{\balpha^1\wedge\balpha^2}\label{covz}
 \end{eqnarray}

where  $q=(q_j)_{j\leq k}$ and $\widetilde{\gamma}=(\widetilde{\gamma}_j)_{j\leq k}$ are defined  in \eqref{seqc} and \eqref{seqq}. Consider a gaussian process $\widetilde{H}_N$ on $\Sigma_{N,k}$   defined by

\be\label{hamrpct}
\widetilde{H}_{N}(\bsigma)\,=\,\frac{1}{\sqrt{N}}\,\sum_{i,j=1}^N\,\widetilde{g}_{ij}(\balpha)\,\sigma_i\sigma_j
\ee

where   $\widetilde{g}_{ij}(\balpha)$ for $i,j=1,\ldots ,N$ are i.i.d. copies of $\widetilde{g}(\balpha)$ in \eqref{covgt}.

Consider also a gaussian process ${G}_N$ on $\Sigma_{N,k}$ independent from $\widetilde{H}_N$ defined by

\be\label{hamrpct}
G_{N}(\bsigma)\,=\,\sum_{i=1}^N\,z_{i}(\balpha)\,\sigma_i
\ee

where   $z_{i}(\balpha)$ for $i=1,\ldots ,N$ are i.i.d. copies of $z(\balpha)$ in \eqref{covz}.
Given two configurations $\bsigma^1=(\sigma^1,\balpha^1),\,\bsigma^2=(\sigma^2,\balpha^2)\in\Sigma_{N,k}$ it's easy to check that the covariances of the process $G_N$ and $H_{N}$ are

\begin{eqnarray}\label{covtree}
 \E \,G_N(\bsigma^1)\,{G}_N(\bsigma^2) &=& 2N \,c_{N,\widetilde{\gamma}}(\bsigma^1,\bsigma^2)\,\widetilde{\gamma}_{\balpha^1\wedge\balpha^2}q_{\balpha^1\wedge\balpha^2}\\
\E \,{\widetilde{H}}_N(\bsigma^1)\,{\widetilde{H}}_N(\bsigma^2) &=& N\,\left(c_{N,\widetilde{\gamma}}(\bsigma^1,\bsigma^2)\right)^2
 \end{eqnarray}

where

\be\label{ctilde}
c_{N,\widetilde{\gamma}}(\bsigma^1,\bsigma^2)=\widetilde{\gamma}_{\balpha^1\wedge\balpha^2}\,q_N(\sigma^1,\sigma^2)
\ee

For $t\in(0,1)$  we define  the interpolating Hamiltonian as

\be\label{hint}
H_{N,t}(\bsigma)\,=\,\sqrt{t}\,{\widetilde{H}}_N(\bsigma)\,+
\,\sqrt{1-t}\,{G}_N(\bsigma)
\ee

and  the interpolating pressure as

\be
\varphi_N(t)= \frac{1}{N}\,\E\,\log Z_{N,t}
\ee

where
\be\label{zrpc}
Z_{N,t}\,=\,\sum_{\bsigma\in\Sigma_{N,k}}\,\nu_{\balpha}\, e^{\,H_{N,t}(\bsigma)}
\ee

The Gibbs measure on $\Sigma_{N,k}$ associated to the Hamiltonian \eqref{hint} is
\be
\mu_{N,t}(\bsigma)\,=\,\dfrac{\nu_{\balpha}\, e^{\,H_{N,t}(\bsigma)}}{Z_{N,t}}
\ee

We denote by $\Omega_{N,t}(\,\cdot \,)$ the average w.r.t. $\mu_{N,t}^{\otimes\infty}$ and by  $\langle\,\cdot\,\rangle_{N,t}$ the quenched expectation $\E\,\Omega_{N,t}(\,\cdot \,)$.

Keeping in mind that  $q_N(\sigma,\sigma)=1$ and $\widetilde{\gamma}_{\balpha\wedge\balpha}=\widetilde{\gamma}_k=\gamma_r$, then using integration by parts formula one obtains

\be\label{sumrule}
2\,\frac{d}{dt}\,\varphi_N\,=\,\widetilde{\gamma}^2_{k}-2\,(\widetilde{\gamma})^2_{k}\,q_k\,+\,\left\langle (\widetilde{\gamma}_{\balpha^1\wedge\balpha^2}q_{\balpha^1\wedge\balpha^2})^2\right\rangle_{N,t}\,-\,\left\langle\Big(c_{N,\widetilde{\gamma}}( \bsigma^1,\bsigma^2) \,-\, \widetilde{\gamma}_{\balpha^1\wedge\balpha^2}\,q_{\balpha^1\wedge\balpha^2}\Big)^2\right\rangle_{N,t}
\ee

Now using the property \eqref{rpcinv} of RPC it's possible to show that
\be
\left\langle (\widetilde{\gamma}_{\balpha^1\wedge\balpha^2}q_{\balpha^1\wedge\balpha^2})^2\right\rangle_{N,t}\,=\, \sum_{j\leq k} \left(\xi_j
-\xi_{j-1} \right)(\widetilde{\gamma}_j q_j)^2\,=\,\widetilde{\gamma}^2_k\,-\,\sum_{0\leq j\leq k-1}\,\xi_j\,\Big(
(\widetilde{\gamma}_{j+1}q_{j+1})^2-(\widetilde{\gamma}_j q_{j})^2\Big)
\ee
In particular \eqref{sumrule} implies that

\be
\varphi_N(1)\leq\varphi_N(0)-\frac{1}{2}\sum_{0\leq j\leq k-1}\,\xi_j\,\Big(
(\widetilde{\gamma}_{j+1}q_{j+1})^2-(\widetilde{\gamma}_{j}q_{j})^2\Big)
\ee

Now since  $H_{N,0}(\bsigma)\,\equiv\,{G}_N(\bsigma)$ it holds

\be
\varphi_N(0)= \frac{1}{N}\,\E\,\log \,\sum_{\balpha\in\N^k}\,\nu_{\balpha}\sum_{\sigma \in\Sigma_{N}}\, e^{\sum_{i=1}^N\,{z}_{i}(\balpha)\,\sigma_i}
\ee

and using again \eqref{rpcinv} one obtains

\be
\varphi_N(0)= \,\E\,\log \,\sum_{\balpha\in\N^k}\,\nu_{\balpha}\,2\cosh\left(z(\balpha)\right)
\ee

Hence

\be
\varphi(0)-\frac{1}{2}\sum_{0\leq j\leq k-1}\,\xi_j\,\Big(
(\widetilde{\gamma}_{j+1}q_{j+1})^2-(\widetilde{\gamma}_{j}q_{j})^2\Big)\, =\,\mathcal{P}_{\beta}(x)
\ee

On the other hand using the recursion in the property \eqref{rpcinv} one can represent $\varphi_N(1)$ in the following way. Let  $\widetilde{H}_N(\sigma,j)$ be a gaussian process on  $(\sigma,j)\in\Sigma_N\times \{1,\ldots,k\}$ with covariance

\be\label{coint}
\E \,\widetilde{H}_N(\sigma^1,j) \widetilde{H}_N(\sigma^2,j') \,=\, N\,\delta_{jj'}\,\left(\sqrt{\,\widetilde{\gamma}^2_j-\widetilde{\gamma}^2_{j-1}\,}\, q_N(\sigma^1,\sigma^2)\right)^2\,
\ee

Then  it holds

\be\label{multiscalep3}
\varphi_N(1)= \frac{1}{N}\,\E \,\log \widetilde{Z}_{0,N}
\ee
\noindent
where  $\widetilde{Z}_{0,N}$ is obtained recursively starting from
\be
\widetilde{Z}_ {k,N}\, =\,\sum_{\sigma}\, \prod_{1\leq j\leq k}\,e^{\,H_N(\sigma,j)}
\ee
\noindent
and for $0\leq j\leq k-1$

\be\label{recZ}
\widetilde{Z}^{\xi_{j-1}}_{j-1,N}\,=\,\E_{j-1}\,\widetilde{Z}^{\xi_{j-1}}_{j,N}
\ee
\noindent
where $\E_j$  averages  the randomness in $H_N(\sigma,j+1)$. \\

Now the key observation is that by definition the sequence $\widetilde{\gamma}$ satisfies

\be\label{seqq2}
\widetilde{\gamma}_j\,=\,\gamma_l\,\,\,\, \mathrm{if}\,\,\, j\in K_l
\ee

If $\widetilde{\gamma}_j=\widetilde{\gamma}_{j-1}$ then  by \eqref{coint} the random variable $\widetilde{H}_N(\sigma,j)$ is actually
a centered gaussian with zero variance, namely its distribution is Dirac delta centered at the origin  and it doesn't play any role. By \eqref{rpcinv3}  $\widetilde{Z}_{0,N}$ can represented using a new Ruelle Probability Cascade  $(\widetilde{\nu}_{\balpha})_{\balpha\in\N^{k-1}}$  that is obtained from
$(\nu_{\balpha})_{\balpha\in\N^k} $  dropping  the point process associated to the intensity $\xi_{j-1}$.
A repeated use of the above argument implies that

\be\label{witloss}
\varphi(1)\,=\, p_N(\beta)
\ee
\noindent

and then we get

\be
p_N(\beta)\leq \mathcal{P}_{\beta}(x)
\ee
\noindent
for every choice of the trial parameter $x\in X_\beta$ and then \eqref{mainl} follows.

\end{proof}

\section{The multi-scale Ghirlanda-Guerra identities}\label{assse}


Consider  quenched pressure density $p_N(\beta)$ in \eqref{multiscalep}. It's standard  to show that

\be\label{ass}
\liminf_{N\to\infty}\,p_N(\beta)\,\geq \liminf_{N\to\infty}\,A_N
\ee

where

\be
A_N\,=\,\E\,\log Z_{N+1}-\E\,\log Z_{N}
\ee

Now the strategy is to compare $Z_{N+1}$ with $Z_N$. This procedure in known in mathematical-physics as Aizenman-Sims-Starr representation \cite{ass, BovierK}. Consider $\rho=(\sigma,\varepsilon)\in \Sigma_{N+1}$
with $(\sigma,\varepsilon)\in\Sigma_{N}\times\{-1,1\}$  then

\be\label{hamass2}
H_{N+1}(\rho,\balpha)\,=\,H'_{N}(\sigma,\balpha)\,+\,\varepsilon\,z_{N}(\sigma,\balpha)+O\left(\frac{1}{N}\right)
\ee

where

\be\label{indgib}
H'_{N}(\sigma,\balpha)\,=\,\frac{1}{\sqrt{N+1}}\,\sum_{i,j=1}^{N}\,g_{ij}(\balpha)\,\sigma_i\sigma_j
\ee

and

\be\label{zass}
z_{N}(\sigma,\balpha)\,=\,\frac{1}{\sqrt{N+1}}\,\sum_{i=1}^{N}\,(g_{i,N+1}(\balpha)\,+\,g_{N+1,i}(\balpha))\,\sigma_i
\ee

On the other hand
\be
H_N(\sigma,\balpha)\,{\buildrel d \over =}\,H'_{N}(\sigma,\balpha)\,+\,y_N(\sigma,\balpha)
\ee

where

\be\label{yass}
y_{N}(\sigma,\balpha)\,=\,\frac{1}{\sqrt{N(N+1)}}\,\sum_{i,j=1}^{N}\,g'_{ij}(\balpha)\,\sigma_i\sigma_j
\ee

for some array $g'$ independent copy of $g$. Given two configurations $\bsigma^1,\bsigma^2\in \Sigma_{N,r}$ the gaussian processes $z_N$ and $y_N$ defined in \eqref{zass} and \eqref{yass} respectively, have covariances

\begin{eqnarray}\label{arrayu}
\E\, z_{N}(\bsigma^1)\, z_{N}(\bsigma^2)\,=\,2\,\frac{N}{N+1}\, \gamma_{\balpha^1\wedge\balpha^2}\,c_{N,\gamma}(\bsigma^1,\bsigma^2)\\
\E\,y_{N}(\sigma,\balpha)\,y_{N}(\sigma',\bbeta)\,=\,\frac{N}{N+1}\,\Big(c_{N,\gamma}(\bsigma^1,\bsigma^2)\Big)^2
\end{eqnarray}

The  above relations implies that

\be\label{cavity}
A_N\,=\,\E\,\log\,\Omega'_N\Big( 2\cosh(z_{N}(\sigma,\balpha)\Big)\,-\, \,\E\,\log\,\Omega'_N\Big( \exp(y_{N}(\sigma,\balpha)\Big)
\ee

where $\Omega'_N=(\omega'_N)^{\otimes\infty}$  and $\omega'_N$ is the Gibbs measure on $\Sigma_{N,r}$ induced by the Hamiltonian $H'_{N}$ in \eqref{indgib}.

The Aizenmann-Sims-Starr representation $A_N$ in \eqref{cavity} for the quenched pressure density   has the same structure of the Parisi functional \eqref{parisifun2}. Hence the strategy is to show that in the thermodynamic limit the distribution of $c_{N,\gamma}(\bsigma^1,\bsigma^2)$  under the random measure $\Omega'_N$ can be well approximated by a suitable RPC.  We have two obstacles to overcomes.

The first problem  is  to understand the joint probability distribution
w.r.t. the limiting Gibbs measure of the two covariances
$q_{N}(\sigma^1,\sigma^2)$ and $\gamma_{\balpha^1\wedge\balpha^2}$. This situation is very similar to the case of the Multispecies SK model \cite{BCMT, panmulti} where it turns out that the Hamiltonian  can be suitably perturbed in order to satisfy a  \textit{synchronization property} that allows to generate the joint probability of different overlaps functions using the same RPC. In addition since the parameter $\xi$ associated to the RPC that express  the Parisi functional  \eqref{parisifun2} satisfies the condition $\zeta\subseteq \xi$, then the same must be true for the one that generates the  limiting distribution of the above overlaps.

In this section we show that the Multiscale SK model can be suitably perturbed in order to satisfy
the \textit{synchronization property}  that actually implies the condition $\zeta\subseteq \xi$.\\

Let $H_N$ be the Hamiltonian function in \eqref{hamrpc} with parameters $\beta=(\zeta,\gamma)$ and $(\nu_{\balpha})_{\balpha\in\N^{r}}$  the random weights of the RPC  associated to the sequence $\zeta$.\\

Let us consider a countable dense subset $\mathcal{W}$ of $[0, 1]^2$ and  a vector
\be
w=(w_s)_{s=0,1}\in\mathcal{W}
\ee

For any $i\in\{0,\ldots,N\}$, $w\in\mathcal{W}$  let us define
\be s_i(w)\,=\,\begin{cases}
\sqrt{N\,w_0 }\, \,\,\,\mathrm{if}\,\,\, i=0\\
\sqrt{w_1}\,\,\,\, \mathrm{otherwise}\,\,\,
\end{cases}
\ee
Let $\left(g^{0}(\balpha)\right)_{\balpha\in\N^r}$ be a family of centered gaussian random variables with covariance
\be\label{g0}
\E \,g^{0}(\balpha^1)\,g^{0}(\balpha^2) = \,\gamma_{\balpha^1\wedge\balpha^2},
\ee

 Consider a gaussian process $h_{N,w,p}$ on $\Sigma_{N,r}$ defined by

\be\label{hamrpc}
h_{N,w,p}(\bsigma)\,=\,\frac{1}{N^{p/2}}\,\sum_{i_1,\ldots,i_p=0}^N\,g^{w,p}_{i_1,\ldots,i_p}(\balpha)\,\sigma_{i_1}s_{i_1}(w)\,\cdots\,\sigma_{i_p}s_{i_p}(w)
\ee

where $\sigma_0=1$ while $g^{w,p}_{i_1,\ldots,i_p}(\balpha)$ for $i_1,\ldots,i_p=1,\ldots ,N$, $p\geq1$ and $w\in\mathcal{W}$ are i.i.d. standard gaussian random variables while if $i_l=0$ form some $1\leq l \leq p$ then $g^{w,p}_{i_1,\ldots,i_p}(\balpha)$ is a family of i.i.d copies of $g^0(\balpha)$ in \eqref{g0} .\\

Then covariance of this process is

\be
\E \,{h}_{N,w,p}(\bsigma^1)\,{h}_{N,w,p}(\bsigma^2) \,=\, \left(R_{N,w}(\bsigma^1,\bsigma^2)\right)^p
\ee
where

\be\label{sinco}
R_{N,w}(\bsigma^1,\bsigma^2)\,=\,w_0\,\gamma_{\balpha^1\wedge\balpha^2}\,+\,w_1\,q_{N}(\sigma^1,\sigma^2)
\ee

We consider a weighted direct sum of the two previous overlaps because in the synchronization mechanism that we are going to exploit we need to control all the terms $\gamma^m q^n$ for generic integers $m$ and $n$.

Since the set $\mathcal{W}$ is countable, we can consider some one-to-one function
$j : \mathcal{W}\rightarrow\N$. Consider now the following gaussian process

\be\label{totalper}
h'_{N}(\bsigma)\,=\,\sum_{w\in\mathcal{W}}\sum_{p\geq1}\,2^{-j(w)-p}\,(\sqrt{\gamma_r+1})^{-p}\,x_{w,p}\,{h}_{N,w,p}(\bsigma)\,
\ee

where $X=(x_{w,p})_{w\in\mathcal{W},\,p\geq1}$ is a family of i.i.d. uniform random variables on $[1,2]$.

Notice that the variance of the process $h'_{N}$ is bounded uniformly on $X$, namely

\be\label{tstb}
\E\, h'_{N}(\bsigma)^2\,\leq\, 4
\ee

For any $\bsigma\in\Sigma_{N,k}$  we define a perturbed  Hamiltonian $H_N^{\mathrm{pert}}$ by

\be\label{hper}
H_N^{\mathrm{pert}}(\bsigma)\,=\,H_N(\bsigma)\,+\,s_N\,h'_{N}(\bsigma)
\ee

where $s_N$ is a sequence of positive real numbers. We start observing that \eqref{tstb} implies that $H_N^{\mathrm{pert}}$ satisfies a \textit{thermodynamic stability} condition

\be\label{thermst}
\E \left(H_N^{\mathrm{pert}}(\bsigma)\right)^2\leq N\,\gamma^2_r+4\,s^2_N
\ee

uniformly on $X$. Consider the random function

\be\label{phip}
\phi_{r,N}\,=\, \log \sum_{\bsigma\in \Sigma_{N,r}}\,\nu_{\balpha} e^{H_N^{\mathrm{pert}}(\bsigma)}
\ee

Then $N^{-1}\,\E\,\phi_{r,N}$ is must be think as a small perturbation and the quantity $p_N(\beta)$  in   \eqref{multiscalep}.  Indeed, it holds

\be
p_N(\beta)\leq \dfrac{1}{N}\,\E\,\phi_{r,N}\leq\, p_N(\beta)\,+\,\dfrac{2s_N^2}{N}\ee

Then if $s_N$ satisfies

\be
\lim_{N\to\infty}\,N^{-1}s_N^2=0
\ee

the thermodynamic limits of $N^{-1}\,\E\,\phi_{r,N}$ and $p_N$ coincide. Moreover RPC concentration inequality given in Proposition \ref{RPCcon}  implies that

\be\label{conc}
\sup\,\left\{\E\,|\phi_{r,N}-\E \phi_{r,N}|\,:\,1\leq\,x_p\leq 2,\,p\geq1 \right\}\leq 4 \,c(\zeta_0)
\ee


for some constant $c(\zeta_0)$ independent of $N$. Hence Theorem 3.2 in \cite{panbook} and  inequality \eqref{conc} implies that if  $s_N=N^\delta$ for $0<\delta<1/2$ we  get the \textit{Multispecies  Ghirlanda-Guerra Identities} (Theorem 2 of \cite{panmulti}) that in our setting reads as follows.\\

Given two configurations $\bsigma^l=(\sigma^l,\balpha^l),\bsigma^{l'}=(\sigma^{l'},\balpha^{l'})\in \Sigma_{N,r}$ we set

\be
R_{l,l'}(w)\,=\,R_{N,w}(\bsigma^l,\bsigma^{l'})
\ee

and

\be\label{ovarray}
R_{l,l'}\,=\,\,\begin{pmatrix} \gamma_{\balpha^{l}\wedge\balpha^{l'}}\\
q_{N}(\sigma^l,\sigma^{l'})
\end{pmatrix}
\ee

Given $n \geq 2$, let
\be
R^n\,=\,\left(R_{l,l'}\right)_{l,l'\leq n}
\ee

and for any bounded measurable function $f = f(R^n)$ we set

\be\label{Gibbsme}
\langle \,f\,\rangle_N\,=\,\E\,\Omega_N (f)
\ee

where $\Omega_N=\mu_N^{\otimes\infty}$ while $\mu_N$ is the random Gibbs measure induced by $H_N^{\mathrm{per}}$
in \eqref{hper}.

For $p\geq 1$ and $w \in\mathcal{W}$ and  conditionally on the i.i.d. uniform sequence $X=(x_{w,p))_{w\in\mathcal{W}},\, p\geq1}$ let
\be
\Delta_N(f, n, w, p,X)\, = \,\Big|\,\langle\,
f\,\left(R_{1,n+1}(w)\right)^p\,\rangle_N
-\frac{1}{n}\,\langle\, f\,\rangle_N\,\langle \,\left(R(w)_{1,2}\right)^p\,\rangle_N
-\frac{1}{n}\,\sum_{l=2}^n\,\langle\,
f\,\left(R_{1,l}(w)\right)^p\,\rangle_N\Big|
\ee

By Theorem 2 in \cite{panmulti} we have that

\be\label{GG}
\lim_{N\to\infty}\,\E_X\, \Delta_N(f, n, w, p,X) = 0
\ee
where $\E_X$  averages the random sequence $X$.

\subsection{The Panchenko's synchronisation property}

The synchronisation property is a powerful tool introduced by Panchenko \cite{panmulti} in his derivation of the lower bound for the multi-specie SK model \cite{BCMT}. It is moreover used in other mean-field settings \cite{Jag,pan1,pan2}.

By Lemma 3.3 in \cite{panbook} there exists a non random sequence $X_N =(x^N_{w,p})_{w\in \mathcal{W} ,\, p\geq1}$ such that \eqref{GG} holds

\be\label{GG2}
\lim_{N\to\infty}\, \Delta_N(f, n, w, p,X_N) = 0
\ee

In the rest of the work we assume  to have  such a sequence $X_N$. Consider the  overlap function

\be\label{gov}
Q_{l,l'}\,=\,\gamma_{\balpha^{l}\wedge\balpha^{l'}}\,+
\,q_{N}(\sigma^l,\sigma^{l'})
\ee
and the following overlap vector

\be
\begin{pmatrix}R^{0}_{l,l'}\\R^{1}_{l,l'}\end{pmatrix}\,=\,\begin{pmatrix} \gamma_{\balpha^{l}\wedge\balpha^{l'}}\\
q_{N}(\sigma^l,\sigma^{l'})
\end{pmatrix}
\ee

Consider also the arrays of the above overlap functions, namely

\be\label{aq}
Q\,=\,\Big(Q_{l,l'}\Big)_{l,l'\geq1}\ee

\be\label{cov2}
\begin{pmatrix} R^0\\R^1\end{pmatrix}\,=\,\begin{pmatrix} R^{0}_{l,l'}\\
R^{1}_{l,l'}
\end{pmatrix}_{l,l'\geq1}
\ee

Let $(N_k)_{k\geq1}$ be any subsequence  along
which the all above overlap arrays  converges
in distribution under the measure $\langle\,\rangle_N$. Since \eqref{GG2} holds, Theorem
3 in\cite{panmulti} implies that the arrays $Q,R^0,R^1$  satisfies the Ghirlanda-Guerra
Identities \cite{GhirlandaGuerra, pan0}, a factorisation property of the quenched equilibrium state (see also \cite{aico,cogigi} for a related factorisation property).

Moreover  Theorem 4 in \cite{panmulti} implies that
the overlaps $R^0$ an $R^1$ are \textit{synchronized}

\begin{prop}\label{sinc}
For for any $s=0,1$ there exists a nondecreasing Lipschitz function

$$L_0 : [0, \gamma_r+1] \longrightarrow [0, \gamma_r],\, L_1 : [0, \gamma_r+1] \longrightarrow [0,1] $$
such that

\be\label{sunc}
R^s_{l,l'}\,=\,L_s(Q_{l,l'})
\ee
almost surely  for all $l,l'\geq1$
\end{prop}
Notice that we can consider the domain and the range of $L_s$ restricted to the positive real line
because each of the overlap arrays  $Q,R^0,R^1$ satisfies the Ghirlanda-Guerra identities and then the  Talagrand's Positivity Principle holds (Theorem 2.16 in \cite{panbook}).

The \textit{synchronization property} of the previous Proposition is already a strong constraint on the limiting overlap distributions. Moreover by construction  the overlap of the Multiscale SK model  has an \textit{apriori}   hierarchical structure encoded in  RPC with parameters $\zeta$.  The combination of these properties implies the following

\begin{prop}\label{cprop}

Let  $F_{Q_{12}}$ be the any weak limit of the  distribution of one element
of the array $Q$, then

\be\label{clpp}
F_{Q_{12}}\in\mathcal{M}_\zeta[0,2\gamma_r]
\ee
where $\mathcal{M}_\zeta[0,2\gamma_r]$ is defined in \ref{funspace}.
\end{prop}

\begin{proof}

The key observation is that  the distribution of $R^{0}_{1,2}$ w.r.t. the perturbed Gibbs measure
$\langle\,\,\rangle_N$ can be exactly computed for any $N$. Indeed by Theorem 3 of \cite{pantal} it holds

\be\label{talpa}
\Big\langle \,\mathds{1}(\balpha^1\wedge\balpha^2=l)\,\Big \rangle_N\,=\,\zeta_l-\zeta_{l-1}
\ee

for any  $l\leq r$ and $N$ integers.

\begin{remark}
The quantity $\dfrac{1}{N}\,\E\,\phi_N$ has a recursive representation analogous to \eqref{multiscalep2}. In particular working with this representation  \eqref{talpa} follows easily using the methods in \cite{broken}.
\end{remark}

By definition $R^{0}_{1,2}=\gamma_{\balpha^1\wedge\balpha^2}$ then

\be\label{key4}
\Big\langle \,\mathds{1}(R^{0}_{1,2}\in A)\,\Big \rangle_N\,=\,\sum_{l=0}^r\,\mathds{1}(\{\gamma_l\}\in A)\,(\zeta_l-\zeta_{l-1})
\ee

for any $N$ and measurable set $A$. Since \eqref{key4} doesn't depends on $N$ the limit along any subsequence
of the distribution of $R^{0}_{1,2}$ w.r.t. $\langle\,\rangle_N$ is given by \eqref{key4}. We denote by   $\langle\,\rangle$ any of the above limiting measure that satisfies the synchronization property  \eqref{sunc}. Hence there exists a  function $L_0$ such that

\be\label{ok5}
R^{0}_{1,2}\,=\,L_0(Q_{1,2}) \,\,\;\mathrm{a.s.}
\ee

For  any $l\leq r$   consider the set

\be\label{al0}
A^{0}_l\,=\,L_0^{-1}\left(\{\gamma_l\}\right)
\ee

Since $L_0$ is nondecreasing Lipschitz then   $A^{0}_l$ is  a closed interval or a single point and

\be\label{che2}
\bigcup_{l\leq r}\, A^{0}_l\,\equiv\,\mathrm{ supp} (Q_{1,2})
\ee

Combining \eqref{key4} and \eqref{ok5} we obtain

\be\label{che}
\Big\langle \,\mathds{1}(Q_{1,2}\in A^0_l)\,\Big \rangle\,=\,\Big\langle \,\mathds{1}(R^{0}_{1,2}=\gamma_l )\,\Big \rangle\,=\,\zeta_l-\zeta_{l-1}
\ee

 If we denote by $Q^-_l$ the left extrema of $A^0_l$ then  \eqref{che} and \eqref{che2} implies that

\be
F_{Q_{12}}(Q^-_l)\,=\,\zeta_{l-1}
\ee

for any $l\leq r$ and this proves the thesis.
\end{proof}

\section{Lower bound}

Let $p_N(\beta)$ be the quenched pressure density of the Multiscale SK model \eqref{multiscalep} and replace the original Hamiltonian $H_N$ with the perturbation  $H_N^{\mathrm{pert}}$ in \eqref{hper}. We already know that this substitution doesn't affect the thermodynamic limit of $p_N(\beta)$. Moreover it entails a small change in the  Aizenmann-Simms-Starr representation given in section \ref{assse}. Indeed
by Theorem 3.6 of \cite{panbook} we have that

\be\label{ass2}
\liminf_{N\to\infty}\,p_N(\beta)\geq \liminf_{N\to\infty}\,\E_X A_N\,+\, o(1)
\ee

where $X= (x_{p,w})_{w\in\mathcal{W},\,p\geq1 }$ is the family of random variables  in \eqref{totalper} and

\be\label{cavity2}
A_N\,=\,\E\,\log\,\Omega_N\Big( 2\cosh(z_{N}(\sigma,\balpha)\Big)\,-\, \,\E\,\log\,\Omega_N\Big( \exp(y_{N}(\sigma,\balpha)\Big)
\ee

Notice that $A_N$ is the same functional appearing in \eqref{cavity} but now $\Omega_N$   is the infinite product of the random Gibbs measure induced by the Hamiltonian $H_N^{\mathrm{pert}}$ in \eqref{hper}.\\

Let us start observing that even if \eqref{cavity2}  is written in average over $X$, Lemma 3.3 of \cite{panbook} ensures that one can  choose a non random sequence $X_N=(x^{(N)}_p)_{p\geq1}$ such that

\be\label{ass2}
\liminf_{N\to\infty}\,p_N(\beta)\geq \liminf_{N\to\infty}\, A_N(X_N)\,+\, o(1)
\ee

and at the same time the \textit{ multi-scale Ghirlanda-Guerra Identities} \eqref{GG} holds.\\

By \eqref{arrayu} and Theorem 1.3 in \cite{panbook} $\,\,A_N(X_N)$ is a continuous functional of the overlap array

\be\label{bji}
\begin{pmatrix}R^{0}_{l,l'}\\R^{1}_{l,l'}\end{pmatrix}_{l,l'\geq1}\,=\,\begin{pmatrix} \gamma_{\balpha^{l}\wedge\balpha^{l'}}\\
q_{N}(\sigma^l,\sigma^{l'})
\end{pmatrix}_{l,l'\geq 1}
\ee
under the measure $\E_{\Omega_N}$. Consider also the array

\be\label{gov}
\big(Q_{l,l'}\big)_{l,l'\geq1}\,=\,\big(\gamma_{\balpha^{l}\wedge\balpha^{l'}}\,+
\,q_{N,\gamma}(\sigma^l,\sigma^{l'})\big)_{l,l'\geq1}
\ee

and  a subsequence $(N_k)_\geq 1$ along which all the above arrays converges in  distribution to some arrays $Q,R^0,R^1$ w.r.t the measure induce by $H_N^{\mathrm{per}}$. By construction the above arrays satisfy    \textit{multi-scale Ghirlanda-Guerra} \eqref{GG} and then   we can apply the results of section \ref{sinc}.

In particular by the \textit{synchronization property} (Proposition \ref{sinc}) for any $s=0,1$ it holds

\be\label{sd}
R^{s}_{l,l'}\,=\,L_s(Q_{l,l'})
\ee

for some nondecreasing Lipschitz function $L_s$. \\

We denote by $\mu_{Q_{12}}$ the distribution of one element of the array $Q$. Let $k\geq 1$ be an integer and consider two sequences $\xi=(\xi_j)_{j\leq k}$ and $c=(c_j)_j\leq k$ such that

\be\label{zxseq}
0=\xi_{-1}<\xi_0< \xi_1<\ldots<\xi_{r-1} < \xi_{k}\,=\,1
\ee

and

\be\label{qseq}
0=c_0<c_1<\ldots< c_k=1+\gamma_{r} \ee

We choose the above couple $(\xi,c)$ such that its associated  discrete distribution $\xi_c$ defined by \eqref{pdis} is close to $\mu_{Q_{12}}$ in some metric that metrizes weak convergence of distributions.

Moreover by Proposition \ref{cprop} we know that the  $F_{Q_{12}}\in\mathcal{M}_\zeta[0,1+\gamma_r]$ then  we can assume without loss that the above $\xi$ satisfies the key property

\be\label{key3}
\zeta\subseteq \xi
\ee

Notice that \eqref{key3} implies that $k\geq r$. Let $(\nu_{\balpha})_{\balpha\in\N^k}$  be the random weights of the RPC associated to $\xi$ in \eqref{zxseq}.
By \eqref{GG}  the array $Q$  satisfies the Ghirlanda-Guerra identities and then Theorems 2.13 and 2.17 in \cite{panmulti} imply that its distribution can be well approximated by the  RPC associated to the above sequences   $\xi$ and $q$ . This means that if we consider a family $(\balpha_l)_{l\geq1}$ of i.i.d.  samples from $\N^{k}$ with distribution given by this RPC we have that the distribution of the  array

\be\label{sample}
(c_{\balpha^l\wedge\balpha^{l'}})_{l,l'}
\ee

will be close to the distribution of the array $Q$. For any $s=0,1$ we define a sequence

\be\label{q12}
q^s_j\,= L_s(c_j)\,\, \,0\leq j\leq k
\ee

then \eqref{sd} implies that for any $s=0,1$ the distribution of the array

\be\label{sample}
Q^s=\,(q^s_{\balpha^l\wedge\balpha^{l'}})_{l,l'}
\ee
will be close to the distribution of the array $R^s$ for any $s=0,1$.

We claim that the triple $(\xi,q^0,q^1)\in X_\beta$ where the set $X_\beta$ is defined in \eqref{spacef}.
In other words we can set $q^0\equiv\widetilde{\gamma}$ and $q^1\equiv q$ for some sequence $\widetilde{\gamma}$ and $q$  in \eqref{seqq}  and \eqref{seqc} respectively.\\

Since  we already know that $\xi$ satisfies \eqref{key3} it's enough to check that $q_0$ in \eqref{q12}
satisfies the condition \eqref{seqq}. For a given $l\leq r$ consider the sets $K_l$ and $A^0_l$
defined in \eqref{mergeset} and \eqref{al0} respectively. Then with probability one
\be\label{al3}
c_j\in A^0_l \Leftrightarrow j\in K_l
\ee

for any $j\leq k$ and any $l\leq r$. Hence, combining  \eqref{q12} and \eqref{al3} we obtain that with probability one if $j\in K_l $ then $q^0_j=\gamma_l$ which coincides with \eqref{seqq}.\\

Given  the above  triple $(\xi,q^0,q^1)\in X_\beta$ consider the Parisi functional $\mathcal{P}(x)$
in \eqref{parisifun2}. Notice that  the quantity $A_N$  \eqref{cavity2} and are $\mathcal{P}(x)$ represented by  the   same continuous functional of the distribution of the arrays
$(R^0,R^1)$ in  \eqref{bji} and $(Q^0,Q^1)$ in \eqref{sample}. Since by construction
these arrays are close in some metric that metrizes weak convergence of distributions
that one can use $\mathcal{P}_{\beta}(x)$ to approximate $A_N(X_N)$ as $N$ goes to infinity (see Section 3.6 in \cite{panbook}).  Hence by \eqref{ass2}

\be\label{final}
\liminf_{N\to\infty}\,p_N(\beta)\geq \inf_{x\in X_\beta}\,\mathcal{P}_{\beta}(x)
\ee
and this conclude the proof of Theorem \ref{upbound}.

In this work we have analysed a multi-scale spin-glass mean-field model and obtained a variational principle that provides the solution for the free energy density. As a bypass result we obtained a full factorisation scheme of ultrametric nature.
We plan to investigate how the multi-scale setting works with other mean-field cases, with hierarchical disordered models \cite{CP,CBG} as well as to extend its use to finite dimensional models where alternative notions of equilibrium state, like for instance the {\it metastate} \cite{SN}, have been advanced.\\ \\

{\bf Acknowledgements}\\
We want to thank several useful discussions with Diego Alberici, Francesco Guerra, Jorge Kurchan and especially Dmitry Panchenko whose observation led to a valuable improvement of proposition \ref{RPCcon}. P.C. was partially supported by PRIN project Statistical Mechanics and Complexity (2015K7KK8L), E.M. was partially supported by Progetto Almaidea 2018.

\section*{Appendix}\label{RPC}
For the benefit of the reader we summarise the main properties of Ruelle probability cascades used in the work. Here we follow Panchenko's monograph on the SK model \cite{panbook}. For the interested reader we also mention  the following works \cite{Ruelle,Bol,Arguin} on RPC and its applications to spin glasses theory.\\

Given an integer $r\geq1$ let $\mathcal{A}=\N^0\cup\N\cup\N^2\ldots\cup \N^r $ be a tree of depth $r$ and root
$\N^0=\{\emptyset\}$. A vertex $\balpha=(n_1,\ldots,n_p)\in\N^p$ for $1<p<r$ has children
$\balpha n= (n_1,\ldots,n_p,n)\in\N^{p+1}$. Therefore each vertex $\balpha=(n_1,\ldots ,n_p)$ is connected to the root
by the path
\be
p(\balpha)=\{n_1,(n_1,n_2), \ldots, (n_1,\ldots,n_p)\}
\ee
We denote by $|\balpha|$ the distance between $\balpha$ and $\empty$ namely the number of coordinates of $\balpha$,
thus by definition $\balpha\in\N^{|\balpha|}$. We also use the notation

\be\label{dwedge}
\balpha\wedge\bbeta=| p(\balpha)\cap p(\bbeta)|
\ee

Let  $\zeta=(\zeta_l)_{l=0,\ldots, r-1}$  be a sequence such that

\be\label{zetaa}
0<\zeta_0< \zeta_1<\ldots<\zeta_{r-1} < 1
\ee
We denote by $(\nu_{\balpha})_{\balpha\in\N^{r}}$ the random weights of the Ruelle probability cascade associated to
the sequence $\zeta$ (Section 2.3 in \cite{panbook}). \\

Consider a family of i.i.d. random variables $\omega= (\omega_p)_{1\leq p\leq r}$ that have the uniform distribution
on $[0,1]$ and some function $X_r=X_r(\omega)$ which satisfies $\E\exp\zeta_{r-1} X_r<\infty$. Let us define
recursively for $0\leq l\leq r-1$
\be\label{rpcrec}
X_l=X_l(\omega_1,\ldots,\omega_l)=\frac{1}{\zeta_l}\log\E_{l}\exp\zeta_l X_{l+1}
\ee

where $\E_{l}$ denotes the expectation with respect to $\omega_{l+1}$.\\

By definition  $X_0$ is not random,  moreover it can be represented trough Ruelle Probability Cascades.
Let $\omega_{\balpha\in\mathcal{A}\setminus\N^0} $ be a family of i.i.d. uniform $[0,1]$ and set
$\Omega_{\balpha}=(\omega_{\bbeta})_{\bbeta\in p(\balpha)}$. Theorem 2.9 in \cite{panbook} reads as follow

\be\label{rpcinv}
 X_0\,=\,\E\log\sum_{\balpha\in\N^r}\nu_{\balpha}\exp X_r\left(\Omega_{\balpha}\right)
\ee

Actually the same argument used in \cite{panbook} to prove \eqref{rpcinv} leads to a remarkable concentration result
for Ruelle probability cascades.

\begin{prop}\label{RPCcon}
For any $r\geq 1$  the random variable\\

\be\label{con1}
\phi_r = \log\sum_{\balpha\in \N^r} \nu_{\balpha} \exp X_r\left(\Omega_{\balpha}\right)
\ee

satisfies

\be\label{selfma}
\E\left(\,\phi_r-\E\, \phi_r\right)^2\, \leq \,4\,c(\zeta_0)
\ee

for some $c(\zeta_0)$ which doesn't depend on the distribution of $X_r$.

\end{prop}

\begin{proof}

Let  $(\nu_{\balpha})_{\balpha\in\N^{r}}$ be the random weights of the Ruelle Probability Cascade associated to the
sequence $\zeta$ in \eqref{zetaa} that we rewrite as

\be\label{conp}
\nu_{\balpha} \,=\,\dfrac{w_{\balpha}}{\sum_{\balpha\in\N^r} w_{\balpha} }
\ee

where the weights $w_{\balpha}$ are defined in section 2.3 of \cite{panbook}.   Let us start with the following lemma.

\begin{lemma}\label{rf}
Let $Z >0$  be a random variable such that $\E Z^{\zeta_{r-1}}<\infty$
and let $(Z_{\balpha} )_{\balpha\in \N^r}$ be a sequence of i.i.d. copies of $Z$ independent of all other random
variables. For any $r\geq 1$  let

\be
Y_r = \log\sum_{\balpha\in \N^r} w_{\balpha} Z_{\balpha}\exp X_r\left(\Omega_{\balpha}\right)
\ee

Then the following holds

\be\label{selfma}
\E\left(\,Y_r-\E\, Y_r\right)^2\, =\,c(\zeta_0)<\infty
\ee

for some $c(\zeta_0)$ which doesn't depend on the distribution of $X_r$ and $Z$.

\end{lemma}

\begin{proof}
The proof is by induction on $r$. Consider the case $r=1$ then

\be
Y_1 = \log\sum_{n\geq1} w_{n} Z_{n}\exp \left(X_1\left(\Omega_{n}\right)\right)
\ee

The invariance property of the Poisson Dirichelet process (Theorem 2.6 in \cite{panbook}) implies that

\be\label{invd}
\sum_{n\geq1} w_{n} Z_{n}\exp \left(X_1(\Omega_{n})\right) \,\,{\buildrel d \over =}\, \,C \sum_{n\geq1} w_{n}
\ee

where $C=\Big(\E (Z \exp(X_1))^{\zeta_0}\Big)^{1/\zeta_0}$. Since

\be
\E\left(\,Y_1-\E\, Y_1\right)^2\, =\,\E\,\left(\,\log\sum_{n\geq 1} w_n Z_{n}\exp \left(X_1(\Omega_{n})\right)-\E\,
\log\sum_{n\geq1} w_n Z_{n}\exp \left(X_1(\Omega_{n})\right) \right)^2
\ee

one can use the invariance property \eqref{invd}  in the r.h.s of the above line obtaining

\be
\E\left(\,Y_1-\E\, Y_1\right)^2\,= \,\E\,\left(\,\log (C\,\sum_{n\geq1} w_n)\, -\,\E\, \log (C\,\sum_{n\geq1} w_n)
\right)^2\,=\,\E\,\left(\,\log\sum_{n\geq1} w_n -\E\, \log\sum_{n\geq1} w_n \right)^2 \ee

Finally  the same argument of Lemma 2.2 in \cite{panbook} implies that

$$\E\,\,\log\sum_{n\geq1} w_n< \infty\,\,\,\, , \,\,\,\E\,\left(\log\sum_{n\geq1} w_n\right)^2 < \infty$$  Therefore we can set
\be
\E\,\left(\,\log\sum_{n\geq1} w_n -\E\, \log\sum_{n\geq1} w_n \right)^2\,=\,c(\zeta_0)<\infty
\ee

for some $c(\zeta_0)$ that doesn't depends on the distribution of $X_1$ and $Z$.\\

Now for an  arbitrary integer $r>1$ consider the quantity

\be\label{pself}
\E\left(\,Y_r-\E\, Y_r\right)^2\,=\,\E\left( \,\log\sum_{\balpha\in \N^r} w_{\balpha} Z_{\balpha}\exp
X_r\left(\Omega_{\balpha}\right)-\E\,\log\sum_{\balpha\in \N^r} w_{\balpha} Z_{\balpha}\exp
X_r\left(\Omega_{\balpha}\right)\right)^2
\ee

The invariance property  (2.57) in \cite{panbook} implies that

\be\label{invd2}
\sum_{\balpha\in \N^r} w_{\balpha} Z_{\balpha}\exp \left(X_r(\Omega_{\balpha})\right) \,\,{\buildrel d \over =}\,
C\,\sum_{\balpha\in \N^{r-1}} w_{\balpha} U_{\balpha}\exp \left(X_{r-1}(\Omega_{\balpha})\right)
\ee

where $C=(\E Z^{\zeta_{r-1}})^{1/\zeta_{r-1}}$ and $U_{\balpha}=\sum_{n\geq1} u_{\balpha n}$
with $\E U_{\balpha}^{\zeta_{r-2}}<\infty$.

Then using \eqref{invd2} in the r.h.s. of \eqref{invd2} we obtain

\be\label{pself}
\E\left(\,Y_r-\E\, Y_r\right)^2\,=\,\E\left( \log\sum_{\balpha\in \N^{r-1}} w_{\balpha} U_{\balpha}\exp
X_{r-1}\left(\Omega_{\balpha}\right)-\E\log\sum_{\balpha\in \N^{r-1}} w_{\balpha} U_{\balpha}\exp
X_{r-1}\left(\Omega_{\balpha}\right)\right)^2
\ee

Finally notice that the above equation is of the same type of \eqref{selfma} with $r$ replaced by $r-1$ and
$Z_{\balpha}$ by $U_{\balpha}$ and this conclude the proof by induction.

\end{proof}

Let's go back to the proof of Proposition \ref{RPCcon}. By \eqref{conp} we can rewrite $\phi_r$ as

\be
\phi_r \,=\, \widetilde{\phi}_r\,+ \log\sum_{\balpha\in \N^r} w_{\balpha}
\ee

where

\be
\widetilde{\phi}_r = \log\sum_{\balpha\in \N^r} w_{\balpha} \exp \left(X_r\left(\Omega_{\balpha}\right)\right)
\ee
Then we can write

\be\label{con2}
\E\left(\,\phi_r-\E\, \phi_r\right)^2\, \leq \,2\,
\E\left(\,\widetilde{\phi}_r-\E\, \widetilde{\phi}_r\right)^2\, + 2 \,\E\,\left(\,\log\sum_{\balpha\in \N^r}
w_{\balpha} \,-\,\E\, \log\sum_{\balpha\in \N^r} w_{\balpha}  \right)^2
\ee

Notice that we can apply Lemma \ref{rf} to compute the two terms in the r.h.s of \eqref{con2}  and this concludes the proof.
\end{proof}

In this work we will use \eqref{rpcinv} also in the following particular setting. Let  $q=(q_l)_{l=0,\ldots, r}$ be a sequence
such that
\be\label{seq3}
0= q_0 < q_1<\ldots<q_{r} < \infty
\ee

and let $(J_l)_{1\leq l\leq r}$ be a family of i.i.d. standard gaussian.

Consider a gaussian random variable

\be\label{hrpc}
H_r =\sum_ {1\leq l\leq r }G_l
\ee

\be\label{Gran}
G_l = J_l\,\left(q_{l}-q_{l-1}\right)^{1/2}
\ee

The covariance of $G$ is given by

\be
\E\,G_l\,G_l'\,=\, \delta_{l,l'} (q_l-q_{l-1})
\ee

Consider the recursive construction \eqref{rpcrec} starting from

\be
X_r\,=\, F(H_r)
\ee

for some function $F$ that satisfies  $\E\exp\zeta_{r-1} X_r<\infty$. Consider  gaussian process $g$ on $\N^{r}$
defined by

\be\label{gtilde}
g(\balpha)\,=\, \sum_{\bbeta\in p(\balpha)}\,J_{\bbeta}\,\left(q_{|\bbeta|}-q_{|\bbeta|-1}\right)^{1/2}
\ee

where $(J_{\balpha})_{\balpha\in\mathcal{A}\setminus \N^0}$ is a family of  i.i.d. standard gaussian random variables.
The covariance of the process $g$ is

\be
\E \,g(\balpha)\,g(\bbeta)\,=\,q_{\balpha\wedge\bbeta}
\ee

Then \eqref{rpcinv} implies that

\be\label{rpcinv2}
 X_0\,=\,\E\log\sum_{\balpha\in\N^r}\nu_{\balpha}\exp F \left(g(\balpha)\right)
\ee

Suppose that instead of \eqref{seq3} we have that $q_l=q_{l-1}$  for some $l\in \{1,\ldots, r\}$.
Then the random variable $G_l$ in \eqref{Gran} is actually a centered gaussian with zero variance, namely
its distribution is a Dirac delta at the origin. This implies that one can set $G_l\equiv 0$ and forget the average
$E_{l-1}$ getting $X_{l-1}=X_l$. In other words $X_0$ can represented using a new Ruelle Probability Cascade
$(\widetilde{\nu}_{\balpha})_{\balpha\in\N^{r-1}}$  that is obtained from
$(\nu_{\balpha})_{\balpha\in\N^r} $  dropping  the point process associated to the intensity $\zeta_{l-1}$.

Formally  we consider the sequence  $\widetilde{\zeta}=\zeta\,\setminus\,\{\zeta_{l-1}\}$  and denote by
$(\widetilde{\nu}_{\balpha})_{\balpha\in\N^{r-1}}$ the random weights of the Ruelle Probability Cascade associated to
the sequence $\widetilde{\zeta}$. Let $\phi$ the one-to-one map between the sets
$\{0,\ldots,r\}\setminus \{l-1\}$ and $\{0,\ldots,r-1\}$ and replace   $H_r$ in \eqref{hrpc}   with
$\widetilde{H}_{r-1} =\sum_ {1\leq l'\neq l\leq r }G_{\phi(l')}$ and  starting from
$\widetilde{X}_{r-1}=F(\widetilde{H}_{r-1})$ we recursively define
\be\label{rpcrec}
\widetilde{X}_{\phi(l')}\,=\,\frac{1}{\zeta_l'}\log\E_{l}\exp\zeta_l' \widetilde{X}_{\phi(l'+1)}
\ee
for any $0 < l'\neq l< r-1$. The it holds

\be\label{rpcinv3}
X_0\,=\,\widetilde{X}_0\,=\,\,\E\log\sum_{\balpha\in\N^{r-1}}\widetilde{\nu}_{\balpha}\exp F
\left(\widetilde{g}(\balpha)\right)
\ee
where $\widetilde{g}(\balpha)$ is defined as in \eqref{gtilde}.

\end{document}